\documentclass[11pt]{article}
\usepackage{amssymb,amsmath,amsthm}
\usepackage{amsfonts}
\usepackage{latexsym}
\usepackage{mathrsfs}
\usepackage{comment,color, marginnote}

\newtheorem{The}{Theorem}
\newtheorem{Pro}{Proposition}
\newtheorem{Rem}{Remark}

\newtheorem{Def}{Definition}
\newcommand{\RR}{\mathbb R}

\def\dq{\dot{q}}
\def\dx{\dot{x}}
\def\dy{\dot{y}}
\def\dz{\dot{z}}
\def\dpx{\dot{p}_x}
\def\dpy{\dot{p}_y}
\def\dpz{\dot{p}_z}

\def\ddq{\ddot{q}}
\def\ddx{\ddot{x}}
\def\ddy{\ddot{y}}

\def\ddr{\ddot{r}}
\def\dr{\dot{r}}
\def\dt{\dot{\theta}}

\date{}
\begin{document}	
\title{Motions of a charged particle in the electromagnetic field induced by a non-stationary current.}
\maketitle
\begin{center}
	\large\author{Manuel Garz\'on\footnote{Departamento de Matem\'atica Aplicada, Universidad de Granada. Facultad de Ciencias, Avenida de Fuentenueva s/n, 18071, Granada, Spain. Email: manuelgarzon@ugr.es}, Stefano Mar\`o\footnote{Dipartimento di Matematica,  Universit\`a di Pisa, Largo Bruno Pontecorvo 5, 56127 Pisa, Italy. Email: stefano.maro@unipi.it}.} 
\end{center}
%\address{Departamento de Matem\'atica Aplicada, Universidad de Granada, Facultad de Ciencias, Granada, Spain}
%\email{manuelgarzon@ugr.es}
%\address{}\foot
%\email{stefano.maro@unipi.it}
\hrule
\begin{abstract}
	In this paper we study the non-relativistic dynamic of a charged particle in the electromagnetic field induced by a periodically time dependent current $J$ along an infinitely long and infinitely thin straight wire. The motions are described by the Lorentz-Newton equation, in which the electromagnetic field is obtained by solving the Maxwell's equations with the current distribution $\vec{J}$ as data. We prove that many features of the integrable time independent case are preserved. More precisely, introducing cylindrical coordinates, we prove the existence of (non-resonant) radially periodic motions that are also of twist type. In particular, these solutions are Lyapunov stable and accumulated by subharmonic and quasiperiodic motions.
\end{abstract}
\hrule
\section{Introduction}
In this paper we consider an electrically neutral infinite straight wire carrying a periodically time dependent current and study the motion of a non relativistic particle (unitary charge and mass) in the corresponding electromagnetic field.  

According to classical electrodynamics (see \cite{Griffiths,Jackson,Planck,Po2}) the motion of the particle in $\RR^3$ is ruled by the Newton-Lorentz equation
\begin{equation}\label{eqLintro1}
\ddq= E(t,q)+\dq\times B(t,q),
\end{equation}
where $E(t,q)$ and $B(t,q)$ represent the electric and magnetic fields generated by the current, being solutions of Maxwell's equations.

Given a current density, finding an explicit formulation for the associated electromagnetic field is a major problem. However, in the case of stationary currents in electrically neutral wires, it turns out that the electric field vanishes and the magnetic field can be computed via the Biot-Savart law. In the case of a constant current along a straight wire, the magnetic field can be computed explicitly, and the magnetic lines are circles around the wire. In this case, equation \eqref{eqLintro1} turns out to be an autonomous integrable Hamiltonian system and the corresponding dynamics have been studied extensively in \cite{ALP,GP1,GP2}.
More precisely, conservation of the energy, linear and angular momenta implies that the particle cannot collide with the wire. Moreover, the motion of particles with non-null angular momentum are helicoidal: radially periodic, turning around the wire and linearly definitely increasing in the direction parallel to the wire. %} \textcolor{red}{Según entiendo, la velocidad de la partícula siempre crece en la dirección de la corriente, pero el crecimiento es lineal cuando se trata del equilibrio. De otro modo tenemos, $\dz=p_z+I_0 \ln (r(t))$, donde $r(t)$ es periódica.}.
In particular, each motion is confined between two cylinders and radially stable.

The introduction of a periodic dependence upon time in the current produces complications both in the computation of the electromagnetic field and in the dynamics given by the Newton-Lorentz equation. First of all, even if there is no charge density, the oscillations of the current generate an electric field and, therefore, the regime is no longer magnetostatic. In particular, Biot-Savart law does not hold in the non stationary case. Solutions of time dependent Maxwell's equations can be rigorously obtained using Jefimenko's formulas for compactly supported current distributions, via the introduction of retarded potentials. In our case the wire does not have compact support being unbounded, but we will prove that the corresponding retarded potential still gives a solution of Maxwell's equations at least in the distributional sense. This is performed by an approximation procedure and gives a rigorous justification of the model we are considering.

From the point of view of dynamics, equation \eqref{eqLintro1} is still Hamiltonian, but time dependent and no more integrable. As a consequence of the periodic dependence, resonances are introduced and collisions with the wire cannot, in principle, be avoided. In general, the motions induced by a time dependent electromagnetic field have not been studied extensively. Some results have been obtained through variational techniques in relativistic regimes in \cite{ABT,ABT2}, where critical point and Lusternik-Schnirelman theories are developed for periodic and Dirichlet boundary conditions. However, in these works, electromagnetic fields are assumed to be regular. This is not satisfied by our problem since collisions with the wire produce singularities. Concerning singular fields, we cite \cite{GT} where periodic solutions are found. However, only isolated singularities are considered, not covering the present case of an infinite wire. This case has been studied in \cite{hau,king,lei_zhang} for a polarized neutral atom under the presence of an electromagnetic field, which is given by a time dependent charge density without current.

In the present paper, we show that several aspects of the integrable dynamics can be recovered in the time dependent case, at least when the current is supposed to be a small perturbation of a constant one. We prove that there exist many motions of the particle with periodic and Lyapunov stable radial component. As a byproduct of our analysis we will also get the existence of solutions that are subharmonic or quasiperiodic in the radial component.

The cylindrical symmetry of the problem is still present in the time dependent case, so that angular and linear momentum are still preserved. Therefore, equation \eqref{eqLintro1} can be reduced to a time dependent Hamiltonian system with one degree of freedom describing the radial component of the solutions. Considering this fact, stability will be understood as Lyapunov stability of the radial component in the whole phase space, not restricting to fixed levels of angular and linear momenta. 

To prove the result we will first work on the reduced system describing the radial component of the motion for fixed values of the momenta. We get the existence of periodic solutions for the time dependent current, through local continuation of equilibria for the (integrable) case of constant current. At this stage we will exclude some resonances and it will be important the assumption of the time dependent current being a small perturbation of a constant one. The stability of the just obtained periodic solutions is proved via the third approximation method, introduced by Ortega in \cite{Ortega} and generalized in \cite{torres_zhang} to a version that will fit in our problem. More precisely, we shall get that the periodic solutions are of twist type, getting also the existence of subharmonic and quasiperiodic solutions. At this stage, it will be necessary to exclude other resonances. This technique has also been used in \cite{lei_zhang} for the case of a time dependent charge density in the wire. Finally, adapting an argument in \cite{SM} we also prove that stability of twist type for fixed values of the momenta implies stability in the whole phase space, allowing variations of the momenta.  

The paper is organized as follows.
In Section \ref{sec:statement} we present the problem in a rigorous way and state our main results in Theorem \ref{mainA} and Theorem \ref{mainB}. 
Section \ref{sec:pot} is dedicated to get a formulation and some properties of the electromagnetic field for the time dependent case, in terms of the vectorial potential. The obtained potential solves Maxwell's equations in a distributional sense. 
The Hamiltonian structure of the problem is discussed in Section \ref{sec:ham} together with the reduction to a time dependent problem with one degree of freedom. Finally, the proofs of our results are given in Section \ref{sec:periodic}.

\section{Statement of the main results}\label{sec:statement}

According to classical electrodynamics, electromagnetic fields are generated by current and charge distributions. More precisely, denoting the charge and current densities by $\rho$ and $\vec{J}$ respectively, the corresponding electric field $E:[0,T]\times\RR^3\to\RR^3$ and magnetic field $B:[0,T]\times\RR^3\to\RR^3$ are the solutions of the Maxwell's equations
\begin{equation}\label{maxwell}
\left\{
\begin{split}
&\nabla\cdot B =0, \quad &\mbox{(Gauss's law for magnetism)} \\
&\nabla\times E + \partial_t B =0, \quad &\mbox{(Maxwell-Faraday equation)} \\
&\nabla\times B = \mu_0\left(\vec{J}+\epsilon_0\partial_t E\right) \quad &\mbox{(Amp\`ere's law)} \\
&\nabla\cdot E =\epsilon_0^{-1}\rho, \quad &\mbox{(Gauss's law)}
\end{split}  
\right.
\end{equation}
where $\epsilon_0, \mu_0$ are the permittivity and permeability of free space. These universal constants are related to the speed of light $c$ in the vacuum by the identity $c\sqrt{\epsilon_0\mu_0}=1$. Moreover, note that here and in the following, the gradient and laplacian symbols $\nabla,\Delta$ will act only on the space variable $q$. A direct consequence of \eqref{maxwell} is the continuity equation:
\begin{equation}
\label{cont}
\frac{\partial \rho}{\partial t}+\nabla\cdot \vec{J} =0.
\end{equation}
On the other hand, the motion of a particle with unitary charge under the presence of an electromagnetic field is ruled by the Newton-Lorentz equation
\begin{eqnarray}\label{eqLintro}
\ddq= E(t,q)+\dq\times B(t,q),\qquad q=(x,y,z)\in\RR^3,
\end{eqnarray}
where the Newtonian approximation is considered, since we will be interested in regimes in which the velocity of the particle $|\dq|$ is small with respect to the speed of light $c$.

We consider an infinitely long straight wire $\mathcal{W}$, which we fix in the $z$-axis without loss of generality, i.e. $\mathcal{W}=\{(0,0,z),\ z\in\RR\}$. We also suppose that the wire is electrically neutral (as many electrons as protons with opposite charge), so that it has no charge density and $\rho=0$. On the other hand, we assume the wire carries an oscillating current with density of the form
\begin{equation}\label{def_J}
J(t):=(I_0+kI(t))\hat{z},
\end{equation}
where $I_0\neq 0$ and $k$ are constants and $I:\RR\to\RR$ satisfies
\begin{equation}\label{hp_I}
I(t+T) = I(t) \mbox{ for every }t\in\mathbb{R}, \qquad\int_0^T I(t)dt =0.
\end{equation}
From a mathematical perspective, it is natural to define the current density $\vec{J} = (J^1,J^2,J^3)$ as the following vectorial distribution in the space-time:
\begin{eqnarray}\label{current}
J^3(f)=\int_{\RR^2}J(t)f(t,0,0,z)dtdz,\ \ J^i(f)=0,\ i=1,2, \ \ \ \forall f\in\mathcal{D}\left(\RR^4\right).
\end{eqnarray}
As usual, $\mathcal{D}\left(\RR^4\right)$ denotes the space of test functions in $\RR^4$, then $\vec{J}\in\mathcal{D}^*\left(\RR^4\right)$ and \eqref{maxwell} must be solved in the sense of distributions. Let us observe that \eqref{current} can be understood as a generalization of the Dirac's Delta for the present situation of a straight thin wire. \\ Finally, regarding the continuity equation, we have
\begin{eqnarray*}
	\nabla\cdot \vec{J} (f) = \partial_z J^3 (f) = J^3(\partial_z f)=0,\qquad \forall f\in\mathcal{D}\left(\RR^4\right),\nonumber
\end{eqnarray*}
then \eqref{cont} is true in the distributional sense.

To state our results, let us introduce cylindrical coordinates $(r,\theta,z)$ in equation \eqref{eqLintro} by
\[
q = (x,y,z) = (r\cos\theta, r\sin\theta, z).
\]
It is well known (see for example \cite{ALP,GP1,GP2}) that in the stationary case corresponding to $k=0$, there exist solutions that move around the wire along cylindrical helices. More generally, the motions of the particle with non-null angular momentum are helicoidal: periodic in the $r$-coordinate, turning around the $z$-axis and linearly increasing in the $z$-coordinate as $t\to\infty$. This comes from the fact that the angular momentum $L=r^2\dot{\theta}$, the linear momentum $p_z=\dot{z}-\ln r$ and the energy $E=\frac{1}{2}(\dr^2+r^2\dot{\theta}^2+\dz^2)$ are first integrals. In particular, the motion is confined between two cylinders of radii $r_1,r_2$ depending on the value of the integrals $L,p_z,E$ and the period of the radial variable depends continuously on these values. 

We are going to show how this picture is modified in the non stationary case $k>0$. To state the results, we note that the energy is no more preserved, while the angular momentum $L =r^2\dot{\theta}$ and the momentum $p_z = \dot{z}+A(t,q)$ are first integral in the present case also (see Section \ref{sec:ham}). Note that in the expression of $p_z$ we have introduced the potential $A(t,q)$ such that $B=\nabla\times A$ (see Section \ref{sec:pot}). Moreover, the constants $c$, $\mu_0$, $\epsilon_0$ do not play any relevant role in the results, so that we normalize them as $\mu_0=2\pi$.

Let us first introduce the kind of solutions that we are going to study.
%For this purpose, we denote by $\Phi_t(q_0,\dot{q}_0)$ the flow of \eqref{eqLintro} and we associate to a solution $q(t)=(r(t),\theta(t),z(t))$ its orbit $\gamma =\{(q(t),\dot{q}(t))\subset\mathbb{R}^6 \: : \: t\in\mathbb{R}\}$.    
\begin{Def}\label{def_sol}
	A solution $q(t)=(r(t),\theta(t),z(t))$ of \eqref{eqLintro} with angular momentum $L$ and linear momentum $p_z$ is called a $(L,p_z)$-solution. Moreover it is
	\begin{itemize}
		\item[{\it (i)}] radially $T$-periodic if $r(t+T)=r(t)$ for every $t\in\mathbb{R}$;
		
		%    \item[{\it(ii)}] radially orbitally stable if for every neighborhood $\mathcal{U}$ of $\{ (r(0),\dot{r}(0),L,p_z): t\in\mathbb{R} \}$ there exists $\mathcal{V}\subset\mathcal{U}$ such that every $(\tilde{L},\tilde{p}_z)$-solution $\tilde{q}(t)$ with $(\tilde{r}(\bar{t}),\dot{\tilde{r}}(\bar{t}),\tilde{L},\tilde{p}_z)\in\mathcal{V}$ for some $\bar{t}\in\mathbb{R}$ satisfies $(\tilde{r}(t),\dot{\tilde{r}}(t),\tilde{L},\tilde{p}_z)\in\mathcal{U}$ for every $t>0$.
		
		\item[{\it(ii)}] radially stable if for every $\epsilon>0$ there exists a neighbourhood $\mathcal{U}$ of $(r(0),\dot{r}(0),L,p_z)$ such that for every $(\tilde{r}_0,\dot{\tilde{r}}_0,\tilde{L},\tilde{p}_z)\in\mathcal{U}$ each $(\tilde{L},\tilde{p}_z)$-solution $\tilde{q}(t)$ with $\tilde{r}(0)=\tilde{r}_0,\dot{\tilde{r}}(0)=\dot{\tilde{r}}_0$ satisfies
		\[
		|\tilde{r}(t)-r(t)|+|\dot{\tilde{r}}(t)-\dot{r}(t)|<\epsilon, \qquad \mbox{for every } t>0.
		\]
		
		\item[{\it (iii)}] radially $(\bar{r},p,q)$-subharmonic, for given $(q,p)\in\mathbb{N}^2$ and $\bar{r}>0$, if $r(t)$ is $qT$-periodic, not $lT$-periodic for every $l=1,\dots,q-1$ and such that the function $r(t)-\bar{r}$ has $2p$ zeros in $[0,qT]$;
		\item[{\it (iv)}] radially (generalized) quasiperiodic with frequencies $(1,\omega)$, for some positive $\omega\in\mathbb{R}$, if the function $r(t)$ is (generalized) quasiperiodic. 
	\end{itemize}
	
\end{Def}
\begin{Rem}
	\begin{enumerate}
		%  \item  Orbital stability is a weaker notion than Lyapunov stability but is the phenomenon expected in Hamiltonian systems (see \cite{SM}).
		\item Condition (ii) means Lyapunov stability in the radial direction w.r.t. small variations of the initial radial coordinates $(r_0,\dot{r}_0)$ and values of the integrals $L,p_z$. Since $L$ and $p_z$ do not depend on the variables $\theta,z$, stability is also guaranteed w.r.t. arbitrary variations of $(\theta_0,z_0)$. 
		\item The standard definition of quasiperiodic solution can be found in \cite{SM}. We decided not to include it to avoid some technicality. We just say that quasiperiodic solutions come in families parametrised by some $\xi\in\mathbb{R}$. The solution is generalized quasiperiodic if the dependence on $\xi$ is not continuous and this depends on the arithmetic properties of the number $\omega$. See \cite{Maro,MO,Ortega96,OrtegaLisboa} for further discussions on it. 
	\end{enumerate}
\end{Rem}

Moreover, we will have to restrict the set of parameters, so that we introduce also the following:

\begin{Def}
	The triplet $(\bar{r},L,p_z)$ with $\bar{r}>0$ and $L\neq 0$ is admissible if 
	\begin{equation*}
	L^2=\bar{r}^2I_0(p_z+I_0\ln(\bar{r})).
	\end{equation*}
	An admissible triplet is non-resonant if
	\begin{equation*}
	T\notin \left\{\frac{n}{\bar{r}}\sqrt{\frac{2L^2}{\bar{r}^2}+I_0^2} , \: n\in\mathbb{N}   \right\}.
	\end{equation*}
	If in addition
	\begin{equation}
	\frac{\sqrt{2L^2+I_0^2\bar{r}^2}}{\bar{r}^2}<\frac{\pi}{2T} \label{stronglynonresonant}
	\end{equation}
	then the admissible resonant triplet $(\bar{r},L,p_z)$ is said strongly non-resonant. 
\end{Def}
\begin{Rem}
	Fixing $\bar{r}>0$ it is easy to show that $L^2=I_0^2\bar{r}^2$, $p_z=I_0-I_0\ln(\bar{r})$ complete an admissible triplet. In this case, the non resonant condition reads as
	\[
	\bar{r}\notin \left\{ n\frac{\sqrt{3}I_0}{T} , \: n\in\mathbb{N}   \right\}
	\]
	and the strong non-resonant condition becomes
	\[
	\frac{\bar{r}}{I_0}>\sqrt{3}\frac{2T}{\pi}.
	\]
\end{Rem}

Now we are ready to state our first result concerning the existence of radially $T$-periodic solutions

\begin{The}\label{mainA}
	Consider a current density $J$ of the form \eqref{def_J}-\eqref{hp_I} with $I\in\mathcal{C}^2([0,T],\mathbb{R})$.\\
	Then, for every admissible non-resonant triplet $(\bar{r},L,p_z)$ there exists a number $k_0>0$ such that, for every $|k|<k_0$, equation \eqref{eqLintro} admits a radially $T$-periodic $(L,p_z)$-solution $(r_k(t),\theta_k(t),z_k(t))$ continuous in $(t,k)$ with $r_0(t)=\bar{r}$ for every $t\in\RR$. Moreover, $\dz(t)=I_0 + \xi_k(t)$ where $\xi_k(t)\to 0$ as $k\to 0$ uniformly in $t\in [0,T]$.
\end{The}

We will also show that under the strong non-resonant condition we can describe the dynamics close to the just introduced solutions.

\begin{The}\label{mainB}
	Consider a current density $J$ of the form \eqref{def_J}-\eqref{hp_I} with $I\in\mathcal{C}^4([0,T],\mathbb{R})$.\\
	Then, for every admissible strongly non-resonant triplet $(\bar{r},L,p_z)$ there exists a number $k_1>0$ such that, for every $|k|<k_1$, the radially $T$-periodic $(L,p_z)$-solution $(r_k(t),\theta_k(t),z_k(t))$ coming from Theorem \ref{mainA} is radially stable. Moreover, there exists $h>0$ such that
	\begin{itemize}
		\item for every $(q,p)\in\mathbb{N}^2$ with $p/q<h$ there exists a radially $(\bar{r},p,q)$-subharmonic $(L,p_z)$-solution;
		\item for every positive irrational $\omega<h$ there exists a radially (generalized) quasiperiodic $(L,p_z)$-solution with frequencies $(1,\omega)$.
	\end{itemize}
	The radial component $r(t)$ of all these solutions converge uniformly to $\bar{r}$ as $k\to 0$.  
	
\end{The}

\begin{Rem}
	The following picture comes from Theorem \ref{mainB}. For every admissible strongly non-resonant triplet $(\bar{r},L,p_z)$ there exist a number $k_1>0$ such that,
	\begin{itemize}
		\item there exist two sequences $r_k^{(0)}\to\bar{r}$ and $\dot{r}_k^{(0)}\to 0$, defined for $|k|<k_1$, such that every $(L,p_z)$-solution $(r_k(t),\theta_k(t),z_k(t))$ of \eqref{eqLintro} with initial position in the cylinder $\mathcal{C}_k=\{(r,\theta,z)\in\mathbb{R}^3: \: r=r_k^{(0)} \}$ and initial radial velocity equal to $\dot{r}_k^{(0)}$ is radially $T$-periodic;
		\item the cylinders $\mathcal{C}_k$ define trapping regions in the following sense. For every $\varepsilon>0$ there exists $\delta$ such that all the $(\tilde{L},\tilde{p}_z)$-solutions with $|L-\tilde{L}|,|p_z-\tilde{p}_z|<\delta$, initial position in the cylinder $\mathcal{C}^\delta_k=\{(r,\theta,z)\in\mathbb{R}^3: \: |r-r_k^{(0)}|<\delta\}$ and initial radial velocity $\dot{r}_0$ satisfying $|\dot{r}_0-\dot{r}_k^{(0)}|<\delta$,  are contained in the cylinder $\mathcal{C}^\varepsilon_k=\{(r,\theta,z)\in\mathbb{R}^3: \: |r-r_k^{(0)}|<\varepsilon \}$ for every time;      
		\item each cylinder $\mathcal{C}_k$ is accumulated by radially subharmonic and radially quasiperiodic solutions.
	\end{itemize}
\end{Rem}

\section{The electromagnetic potential.}\label{sec:pot}

In this section, the electric and magnetic fields generated by the current density $\vec{J}$ in \eqref{def_J} are deduced. Here we will not normalize the constants $c$, $\mu_0$, $\epsilon_0$.

Formally, by applying Helmholtz's theorem to the first two equations in \eqref{maxwell}, the electromagnetic field is expressed through the corresponding scalar and vectorial potentials  $\Phi:[0,T]\times\RR^3\to\RR$ and $A:[0,T]\times\RR^3\to\RR^3$ as
\begin{equation}\label{gen_pot}
E=-\nabla \Phi-\frac{\partial A}{\partial t},\qquad
B=\nabla\times A.
\end{equation}
Notice that these potentials are not unique. In fact, for any field like \eqref{gen_pot} there exists an infinite class of potentials that generates it. This is known as the \textsl{Gauge Invariance} of Maxwell's equations, and it is consistent since potentials are not physical observables. However, the generated electromagnetic field via \eqref{gen_pot} is the unique solution of \eqref{maxwell}.

We assume that $\Phi$ and $A$ satisfy the Lorenz Gauge:
\begin{equation*}
\partial_t\Phi+c^2\nabla\cdot A =0.
\end{equation*}
This allows to uncouple the system and then \eqref{maxwell} is reduced to the following wave equations:
\begin{equation*}
\left\{
\begin{split}
&\partial^2_tA -c^2\Delta A=\dfrac{\vec{J}}{\epsilon_0}, \\
&\partial^2_t\Phi -c^2\Delta \Phi=\dfrac{\rho}{\mu_0\epsilon_0^2}.
\end{split}
\right. 
\end{equation*}
Since our case is $\rho =0$, we can take $\Phi =0$ so that we are left to
\begin{equation}\label{pot_max}
\left\{
\begin{split}
&\partial^2_tA -c^2\Delta A=\dfrac{\vec{J}}{\epsilon_0}, \\
&\nabla\cdot A =0.
\end{split}
\right. 
\end{equation}

If the data is a compactly supported distribution, i.e. if it belongs to $\mathcal{D}'\left(\RR^4\right)$, the solution can be directly obtained by convolution with the fundamental solution of the wave operator in $\RR\times\RR^3$:
\begin{eqnarray}
\label{FundSolWave}
\mathcal{F}(t,x)=H(t)\dfrac{\delta_{|x|}(ct)}{4\pi c|x|}.     
\end{eqnarray}
Here, $\delta_{|x|}(ct)$ represents the Dirac's delta in the radial variable at the point $ct$ and $H(t)$ is the Heaviside function. In this way, \eqref{pot_max} is solved by the retarded potential
\begin{align}
% V(t,q)&= \frac{1}{4\pi\epsilon_0}\int \frac{\rho(q',t_r)}{|q-q'|}dq',\\
\label{Aret}
A(t,q)&=\frac{\mu_0}{4\pi}\int_{\RR^3} \frac{\vec{J}(q',t_r)}{|q-q'|}dq',
\end{align}
where
\[
t_r = t -\frac{|q-q'|}{c}
\]
is called retarded time and takes into account the finiteness of the speed of light.

However, since the wire $\mathcal{W}$ is unbounded, then $\vec{J}$ is not compactly supported and formula \eqref{Aret} needs some justification. Actually, using \eqref{def_J} we get
\[
A(t,r)= \frac{\mu_0}{2\pi}\left[\int_0^\infty\frac{I_0}{\sqrt{\tau^2 + r^2}} d\tau +k\int_0^\infty\frac{I\left(t-\frac{\sqrt{r^2+\tau^2}}{c}\right)}{\sqrt{\tau^2 + r^2}} d\tau\right]\hat{z}
\]
and the first integral is divergent. This is not contradictory with the well known case $k=0$. In fact, in this case one can compute the magnetic field $B(q)$ via the Biot-Savart formula and get that $A(r)=-\frac{\mu_0}{2\pi}\ln r$ is a potential by checking directly that $B=\nabla\times A$. 
In the non stationary case, everything is consistent in the sense of distributions, as discussed in the following proposition. To state it, we use the notation
\begin{eqnarray*}
	f[t,r,\tau]:=f\left(t-\frac{\sqrt{r^2 + \tau^2}}{c}\right).
\end{eqnarray*}

%With this expression, one get the Jefimenko equations for $E$ and $B$:
%\begin{align}
%  \label{EJef}
%  E(t,q) &= \frac{1}{4\pi\epsilon_0} \int \left[ \left( \frac{\rho(q',t_r)}{|q-q'|^3}+\frac{1}{c|q-q'|^2} \frac{\partial\rho(q',t_r)}{\partial t}\right)(q-q')-\frac{1}{c^2|q-q'|}\frac{\partial J(q',t_r)}{\partial t} \right] dq' \\
%  \label{BJef}
%B(t,q) &= \frac{\mu_0}{4\pi}\int \left(\frac{J(q',t_r)}{|q-q'|^3}+\frac{1}{c|q-q'|^2}\frac{\partial J(q',t_r)}{\partial t}\right)\times(q-q')   dq'  
%  \end{align}
\begin{Pro}\label{lem_1}
	Suppose that $I(t)$ is of class $\mathcal{C}^n,$ with $n\geq 1$, and satisfies \eqref{hp_I}.\\
	Then,
	\begin{equation*}
	A(t,q) = -\frac{\mu_0}{2\pi}[a_0(r)+ka(t,r)] \hat{z}%, \quad a(t,r) = \frac{\mu_0}{2\pi}\int_r^\infty \frac{I_0+kI(t-z/c)}{\sqrt{\rho-r^2}} dz 
	\end{equation*}
	where
	\begin{equation}
	\label{defa}
	a_0(r) = I_0\ln r , \quad  a(t,r)=  \int_0^\infty\frac{I[t,r,\tau]}{\sqrt{\tau^2 + r^2}} d\tau.
	\end{equation}
	is a solution of \eqref{pot_max} in the distributional sense.
	
	Moreover, the function $a(t,r)$ in \eqref{defa}: 
	\begin{itemize}
		\item is well defined and belongs to $\mathcal{C}^n\left(\RR\times\RR^+\right) $;
		\item is $T$-periodic in $t$.
	\end{itemize}
\end{Pro}

\begin{proof}
	We begin studying the properties of $a(t,r)$. Firstly, to prove that the function is well defined, let us see that the integral converges writing it as
	\[
	a(t,r)=\int_0^\infty \frac{I[t,r,\tau]}{\sqrt{\tau^2 + r^2}} d\tau = \int_0^{r} \frac{I[t,r,\tau]}{\sqrt{\tau^2 + r^2}} d\tau+ \int_{r}^\infty \frac{ I[t,r,\tau]}{\sqrt{\tau^2 + r^2}} d\tau.
	\]
	So, the first integral is finite since
	\[
	\left|\int_0^{r} \frac{ I[t,r,\tau]}{\sqrt{\tau^2 + r^2}} d\tau \right|\leq \|I\|_{\infty}\int_0^{r} \frac{1}{\sqrt{\tau^2 + r^2}} d\tau = \|I\|_{\infty} \ln\left(1+\sqrt{2}\right)
	\]
	and $I$ is bounded. For the second, through integration by parts and denoting $\mathcal{I}(t)$ as a primitive of $I(t)$ we get
	\[
	\int_{r}^\infty \frac{ I[t,r,\tau]}{\sqrt{\tau^2 + r^2}} d\tau = -c\left[ \frac{\mathcal{I}[t,r,\tau]}{\tau}\right]^\infty_{r}+c\int_{r}^\infty \frac{\mathcal{I}[t,r,\tau]}{\tau^2} d\tau. 
	\]
	Since $I(t)$ has null average, then $\mathcal{I}(t)$ is periodic, so that
	\begin{eqnarray}
	\left|\int_{r}^\infty \frac{ I[t,r,\tau]}{\sqrt{\tau^2 + r^2}} d\tau\right| &=&\left|-c\left[ \frac{\mathcal{I}[t,r,\tau]}{\tau}\right]^\infty_{r}+c\int_{r}^\infty \frac{\mathcal{I}[t,r,\tau]}{\tau^2} d\tau\right|\nonumber\\
	&\leq&c\frac{\left|\mathcal{I}(t-\sqrt{2}r/c)\right|}{r} +c\left|\int_{r}^\infty \frac{\mathcal{I}[t,r,\tau]}{\tau^2} d\tau\right|\leq\dfrac{2c\|\mathcal{I}\|_\infty}{r}.\label{locallyintegrable}
	\end{eqnarray}
	Hence, $a(t,r)$ is well defined for $(t,r)\in[0,T]\times(0,+\infty)$. Let us now consider the regularity, by proving that $a(t,r)$ is $\mathcal{C}^k$ on $[0,T]\times(a,b)$ for every $0<a<b<+\infty$. \\
	The integrand function in $a(t,r)$ is continuous for $(t,r,\tau)\in [0,T]\times[a,b]\times[0,+\infty)$ hence, to prove the continuity of $a(t,r)$ is enough to prove that the improper integral is uniformly convergent on $[0,T]\times[a,b]$. This follows integrating by parts, actually.
	\begin{eqnarray}\label{unif}
	\left|\int_{m}^{+\infty}\dfrac{I\left[t,r,\tau\right]}{\sqrt{ \tau^2 + r^2}}d\tau \right| = \left|c\left[ \frac{\mathcal{I}[t,r,\tau]}{\tau}\right]^\infty_{m}+c\int_{m}^\infty \frac{\mathcal{I}[t,r,\tau]}{\tau^2} d\tau\right|\leq \dfrac{2c\|\mathcal{I}\|_\infty}{m}.
	\end{eqnarray}
	Concerning the derivatives, let us denote by $F_{ij}(t,r,\tau)$ the derivatives of the integrand function w.r.t $(t,r)$ of order $(i,j)$ such that $i+j\leq n$. It can be checked that they are all continuous in $(t,r,\tau)\in [0,T]\times[a,b]\times[0,+\infty)$. As before, we need to check that the corresponding improper integrals converge uniformly. This can be seen by the following observations. First of all, by considering the cases $F_{i0}$ (i.e. the derivatives involve the sole variable $t$) the uniform convergence is proved as in  \eqref{unif}, since
	\[
	F_{i0}(t,r,\tau) = \dfrac{I^{(i)}\left[t,r,\tau\right]}{\sqrt{ \tau^2 + r^2}}.
	\]
	On the other hand, if a derivative involves the variable $r$, fixing $t$ and $r$ we have:
	\[
	F_{ij}(t,r,\tau) \sim O\left(\frac{1}{\tau^2}\right) \qquad j>0,
	\]
	as $\tau\to\infty$. Therefore, uniform convergence follows from the fact that $r$ belongs to a compact set. Finally, the periodicity in $t$ of $a(t,r)$ follows from the periodicity of $I$.
	
	Note that estimate \eqref{locallyintegrable} does not imply that $a(t,r)$ is a locally integrable function in $r$, however, it allows to define it as a distribution in $\RR^4$. As usual, considering $f\in\mathcal{D}\left(\RR^4\right)$, we establish $a(f)$ as the scalar product in $L^2\left(\RR^4\right)$ of $a(t,r)$ with $f(t,q)$, where $r$ denotes the radial component of $q$ in cylindrical coordinates:
	\begin{eqnarray}
	a(f)&=&\int_{\RR^4}f(t,q)\int_{0}^{r}\dfrac{I[t,r,\tau]}{\sqrt{\tau^2+r^2}}d\tau dtdq+\int_{\RR^4}f(t,q)\int_{r}^{\infty}\dfrac{I[t,r,\tau]}{\sqrt{\tau^2+r^2}}d\tau dtdq\nonumber\\
	&\leq&\ln\left(1+\sqrt{2}\right)\|I\|_\infty\|f\|_{L^1} + 2c\|\mathcal{I}\|_\infty\int_{\RR^4}\dfrac{|f(t,q)|}{r}dtdq\nonumber.
	\end{eqnarray}
	Passing to cylindrical coordinates, we obtain the bound
	\begin{eqnarray}
	a(f)&\leq&2cM(f)\|f\|_{L^\infty}\left(\|I\|_\infty+\|\mathcal{I}\|_\infty\right),\nonumber
	\end{eqnarray}
	where $|\text{supp}(f)|$ denotes the Lebesgue-measure of the support of $f$ and $M(f)$ is a constant that depends of $\text{supp}(f)$.
	
	Now we deduce the potential generated by the wire current distribution. As a first step, we separate the expression \eqref{current} as the sum of the constant and non constant part, i.e. $J=J_0+J_k$ such as
	\begin{eqnarray}
	J_0(f)=\int_{\RR^2}I_0f(t,0,0,z)dtdz,\ J_k(f)=k\int_{\RR^2}I(t)f(t,0,0,z)dtdz,\ \ \forall f\in\mathcal{D}\left(\RR^4\right).\nonumber
	\end{eqnarray}
	Here we denoted by $J$ the third component of \eqref{current} being null the other components. Note that by linearity of Maxwell's equation the solution for $J$ can be expressed as the sum of the particular solutions for $J_0$ and $J_k$. It is well known that $a_0(r)$, with $r^2=x^2+y^2$, solves \eqref{pot_max} for $J_0$. That is, for any $f\in\mathcal{D}\left(\mathbb{R}^4\right)$,
	\begin{eqnarray}
	-\Delta a_0(f)&=&I_0\int_{\RR^4}\ln(r)\Delta f(t,q)dtdq=I_0\int_{\RR^4}\ln(r)\Delta_{x,y} f(t,q) dtdq \nonumber\\&=&2\pi I_0\int_{\RR^2}f(t,0,0,z)dtdz=2\pi J_0(f),\nonumber
	\end{eqnarray}
	where we have used that the logarithm is the fundamental solution for the Laplace's operator in $\RR^2$. On the other hand, we cannot compute directly the solution for $J_k$ as a convolution with \eqref{FundSolWave}, because the data is not of compact support, as it was said before. For that reason, we define the sequence $\{J_{k,m}\}_m$ in $\mathcal{D}'\left(\RR\right)$ as
	\begin{eqnarray}
	J_{k,m}(f)=k\int_{[-m,m]^2}f(t,0,0,z)J(t)dtdz\nonumber,
	\end{eqnarray}
	which clearly converges to $J_k$ when $m\nearrow\infty$. Fixed $k>0$, we also denote by $a_m$ the corresponding solution of the wave equation. By computing \eqref{Aret} we can deduce the expression
	\begin{eqnarray}
	a^{m}(t,r,z)=\dfrac{\mu_0}{4\pi}\int_{-m}^{m}\dfrac{I\left(t-c^{-1}\sqrt{r^2+ (z'-z)^2}\right)}{\sqrt{r^2+ (z'-z)^2}}dz'.\nonumber
	\end{eqnarray}
	Again, by the change of variable $\tau=z-z'$ and integrating by parts, it is proved that $a^m$ converges uniformly to $a$ when $m\nearrow\infty$, which implies its convergence in the distributional sense. Due to this, the function $a(t,r)$ solves the wave equation for the data $J_k$. To conclude, we have to see that the retarded potential $A(t,r)$ defined in \eqref{Aret} is a Lorenz Gauge, and then it is a solution of Maxwell's equations. But this is trivial because the two first components of $A$ and the scalar potential $\Phi$ are zero. So, we only has to compute the derivative $\partial_z A$, which is clearly null.

\end{proof}

\begin{Rem}
	In case $I(t) = \sin(t)$ the function $a(t,r)$ has an explicit representation. 
	Actually, performing the change of variable $\tau^2+r^2 = r^2\cosh^2\xi$
	\begin{align*}
	a(t,r) &= \int_0^\infty \frac{\sin[t,r,\tau]}{\sqrt{\tau^2+r^2}} d\tau = \int_0^\infty \sin\left(t-\frac{r}{c}\cosh\xi  \right) d\xi \\
	&=\frac{\pi}{2}\left[\sin t \, \frac{2}{\pi}\int_0^\infty \cos\left(\frac{ r}{c}\cosh\xi  \right)d\xi-\cos t \, \frac{2}{\pi}\int_0^\infty \sin\left(\frac{ r}{c}\cosh\xi  \right)d\xi   \right] \\
	&= \frac{\pi}{2}\left[ \sin t \, Y_0\left(\frac{ r}{c}\right)+\cos t \, J_0\left(\frac{ r}{c}\right)    \right]
	\end{align*}
	where $J_0(x),Y_0(x)$ are the Bessel functions of order zero of first and second kind respectively
\end{Rem}

\section{Hamiltonian formulation and reduction}\label{sec:ham}

From Proposition \ref{lem_1}, system \eqref{eqLintro} reduces to 
\begin{eqnarray*}
	\ddq= -\frac{\partial A}{\partial t}(t,q)+\dq \times\left( \nabla\times A(t,q)\right),    \qquad q=(x,y,z)\in\RR^3,
\end{eqnarray*}
that is Hamiltonian with
\[
H(t,q,p) = \frac{1}{2}\left|p-A(t,q) \right|^2.
\]
From $\dq = \partial H/\partial p$ we get that the momenta $p=(p_x,p_y,p_z)$ are
\[
p = \dq+A(t,q).
\]

With these assumptions, and normalizing $\mu_0 = 2\pi$ the Hamilton equations read
\[
H = \frac{1}{2}\left[ p^2_x+p_y^2+ \left( p_z+ I_0\ln r + k a(t,r) \right)^2  \right]
\]
and
\begin{equation*}
\left\{
\begin{split}
\dpx &= -(p_z+ I_0\ln r + k a(t,r) ) \left(\frac{I_0}{r^2}+k\frac{\partial_ra(t,r)}{r}\right)x        \\
\dpy &=  -(p_z+ I_0\ln r + k a(t,r) ) \left(\frac{I_0}{r^2}+k\frac{\partial_ra(t,r)}{r}\right)y    \\
\dpz &= 0     \\
\dx &= p_x      \\
\dy &= p_y     \\
\dz &= p_z+ I_0\ln r + k a(t,r)      
\end{split}
\right.
\end{equation*}
From this,
\begin{equation*}
\left\{
\begin{split}
\ddx &=- (p_z+ I_0\ln r + k a(t,r) ) \left(\frac{I_0}{r^2}+k\frac{\partial_ra(t,r)}{r}\right)x \\
\ddy &= -(p_z+ I_0\ln r + k a(t,r) ) \left(\frac{I_0}{r^2}+k\frac{\partial_ra(t,r)}{r}\right)y \\
\dz &= p_z+ I_0\ln r + k a(t,r)
\end{split}
\right.
\end{equation*}
and the momentum $p_z$ is a first integral. Passing in polar coordinates, $(x,y) = r(\cos\theta,\sin\theta)$, we get that the norm of the angular momentum is another first integral:
\[
\dy x-\dx y = r^2\dt =L.
\]
Differentiating twice the equation $r^2=x^2+y^2$ and noting that $\dx^2+\dy^2 = \dr^2+r^2\dt^2$ we get
\begin{equation}
\label{eqddr}
\ddr = \frac{L^2}{r^3}-(p_z+ I_0\ln r + k a(t,r) ) \left(\frac{I_0}{r}+k \partial_ra(t,r)\right).
\end{equation}
Hence, we will consider the equations
\begin{equation}\label{eq_cil}
\left\{
\begin{split}
\ddr & = \frac{L^2}{r^3}-(p_z+ I_0\ln r + k a(t,r) ) \left(\frac{I_0}{r}+k \partial_ra(t,r)\right), \\
r^2\dt & = L, \\
\dz &= p_z+ I_0\ln r + k a(t,r).
\end{split}
\right.
\end{equation}
Note that $\ddr = -\partial_rV(t,r)$ with
\begin{equation}\label{defV}
V(t,r) = \frac{L^2}{2r^2}+\frac{1}{2}(p_z+ I_0\ln r + k a(t,r) )^2.
\nonumber\end{equation}
If $k=0$, then also the energy
\[
E=\frac{1}{2}\dr^2 +V_0(r) = \frac{1}{2}\dr^2 + \frac{L^2}{2r^2}+\frac{1}{2}(p_z+ I_0\ln r)^2 
\]
is preserved. If $L\neq 0$ then $V_0(r)$ has only one minimum $\bar{r}$  and $V_0(r)\to\infty$ as $r\to 0,+\infty$. Hence every solution of \eqref{eqddr} is periodic. Using the solution $r(t)=\bar{r}$ in system \eqref{eq_cil} we get $\dt(t)= L/\bar{r}^{2}$ and $\dz(t) = p_z+I_0\ln\bar{r}$. Consequently, the particle moves on a cylindrical helix when $L\neq 0$ and $\dz(0)= p_z+I_0\ln\bar{r} \neq 0$.    

We want to study how this dynamics is perturbed for $k>0$.

\begin{Rem}
	In \cite{lei_zhang} it is considered the induced motion of a polarized neutral atom (i.e. a dipole) under the electromagnetic field generated by a time dependent charge density (and no current) in the wire $\mathcal{W}$. As in our case, the motion can be reduced to the radial component giving a singular equation of the form
	\[
	\ddot{r}=\frac{L^2-A+\rho(t)}{r^3}
	\]
	where $L$ is the angular momentum, $A$ a constant depending on the physical properties of the atom and $\rho(t)$ represents the charge density. Note that this equation and \eqref{eqddr} have different singularities. 
\end{Rem}

\section{Existence of solutions with periodic radial oscillations of twist type}\label{sec:periodic}
In this section we apply the local continuation Theorem to equation \eqref{eqddr}. For the sake of convenience we recall here the classical theorem taken from \cite{Cod}. To this aim, let us consider the differential equation
\begin{equation}\label{eq0}
\dx = f(t,x,\mu)
\end{equation}
where $f:V \times B(\mu_0)\rightarrow\mathbb{R}^n$
having denoted by $V$ a domain of $\mathbb{R}\times\mathbb{R}^n$ and by $B(\mu_0)$ the open ball of radius $\mu_0$ in $\mathbb{R}^m$. The function $f$ is continuous in $(t,x,\mu)$ and has first-order derivatives w.r.t. the components $x_i$ of $x$.

Suppose that $f(t+T,x,\mu)=f(t,x,\mu)$ and that for $\mu=0$, equation \eqref{eq0} admits a $T$-periodic solution $p(t)$ such that $(t,p(t))\in V$ for every $t$. We also introduce the first variation of \eqref{eq0} w.r.t. the solution $p(t)$ as
\begin{equation}\label{var_eq}
\dy = f_x (t,p(t),\mu) y
\end{equation}
where $f_x$ denotes the Jacobian matrix of $f$ w.r.t. the variable $x$. Under these assumptions, we can state 

\begin{The}[\cite{Cod}]\label{teo_cont}
	Suppose that the first variation \eqref{var_eq} for $\mu=0$ has no solution of period $T$. Then, there exists $\mu_1<\mu_0$ such that for every $|\mu|<\mu_1$, equation \eqref{eq0} has a unique $T$-periodic and continuous solution $q=q(t,\mu)$, such that $q(t,0)=p(t)$. 
\end{The}

\begin{Rem}\label{rem_teo}
	It comes from the proof of Theorem \ref{teo_cont} that, for every $|\mu|<\mu_1$, the functions $q(t,\mu)$ are solutions of \eqref{eq0} with initial condition $x(0)=\alpha(\mu)$ where $\alpha(\cdot)$ is a continuous function defined in a neighbourhood of $0$ such that $\alpha(0)=p(0)$. Since $p(t)$ is a solution of \eqref{eq0} for $\mu=0$ and initial condition $x(0)=p(0)$, by continuous dependence w.r.t. to parameters (see \cite[Chapter V, Theorem 2.1]{Hart}) we have that actually
	\[
	\lim_{\mu\to 0}q(t,\mu) = p(t) \quad\mbox{uniformly in } t\in[0,T]. 
	\]
\end{Rem}

We want to prove the following result
\begin{Pro}\label{prop_cont}
	Suppose that $I\in\mathcal{C}^2([0,T],\mathbb{R})$ and satisfies \eqref{hp_I}. For every admissible non-resonant triplet $(\bar{r},L,p_z)$ there exists a positive $k_0$ such that, for every $|k|<k_0$, equation \eqref{eqddr} admits  a unique positive $T$-periodic $(L,p_z)$-solution $r_k(t)$ that is continuous in $(t,k)$ and such that $r_k(t)\rightarrow\bar{r}$ as $|k|\to 0$ uniformly in $t\in[0,T]$.
\end{Pro}

\begin{proof}
	
	Let us fix an admissible non-resonant triplet $(\bar{r},L,p_z)$ and consider only the case $I_0>0$, being the other case similar. In order to apply Theorem \ref{teo_cont} to equation \eqref{eqddr} we note that it can be written in the form \eqref{eq0} as a system of first order and the regularity assumptions are satisfied in a neighbourhood of $\bar{r}$ by Proposition \ref{lem_1}. For $k=0$ equation \eqref{eqddr} reduces to    
	\begin{equation*}
	\ddot r = \frac{L^2}{r^3}-\frac{I_0p_z}{r}-\frac{I_0^2\ln(r)}{r}
	\end{equation*}
	that admits the constant solution $r(t)=\bar{r}$ since $\bar{r},L,p_z$ belong to an admissible triplet.
	
	Writing $p_z=p_z(\bar{r},L)$, the first variation w.r.t. $r(t)=\bar{r}$ for $k=0$ of \eqref{eqddr} can be written as
	\begin{eqnarray*}
		\dy
		%&=&
		%\begin{pmatrix}
		% 0&1\\-2\dot{\theta}^2-I_0^2r^{-2}&0
		%\end{pmatrix}y
		%\nonumber\\
		&=&
		\begin{pmatrix}
			0  &   1  \\
			-\frac{1}{\bar{r}^{2}}\left(\frac{2L^2}{\bar{r}^{2}} + I_0^2\right) &  0 
		\end{pmatrix}y, 
	\end{eqnarray*}
	whose solutions are $T_0$-periodic with
	\[
	T_0=\frac{1}{\bar{r}}\sqrt{\frac{2L^2}{\bar{r}^2}+I_0^2}.
	\]
	Since by hypothesis $T\notin\{nT_0, \: n\in\mathbb{N}\}$, Theorem \ref{teo_cont} and Remark \ref{rem_teo}, give the existence of $k_0>0$ such that for every $|k|<k_0$ there exists a unique $T$-periodic solution $r_k(t)$ of \eqref{eqddr}, and $r_k(t)\rightarrow \bar{r}$ as $k\to 0$ uniformly in $t$. Finally, by the uniform convergence to $\bar{r}>0$ one can eventually decrease the value of $k_0$ in order to guarantee that, for every fixed $|k|<k_0$, $r_k(t)>0$ for every $t\in [0,T]$.
\end{proof}

Now, it is straightforward the

\begin{proof}[Proof of Theorem \ref{mainA}]
	Fixing an admissible non-resonant triplet $(\bar{r},L,p_z)$, Proposition \ref{prop_cont} guarantees the existence of a positive solution $r_k(t)$ of \eqref{eqddr} continuous in $(t,k)$ such that $r_k(t+T)=r_k(t)$, $r_0(t)=\bar{r}$. Inserting this solution in \eqref{eq_cil} we get the thesis. Actually, the angular momentum $L$ is not zero. Concerning the sign of $\dz$ we have
	\[
	\dz = I_0 +I_0[\ln(r_k(t))-\ln(\bar{r})]+ka(t,r_k(t)). 
	\]
	Hence, since $r_k(t)\to \bar{r}$ uniformly in $t$ as $k\to 0$ and $a(t,r)$ is continuous in a neighbourhood of $\bar{r}$, we get the result for some eventually smaller $k_0$ recalling that $I_0\neq 0$.
\end{proof}

\subsection{Proof of Theorem \ref{mainB}}\label{sec:twist}
We shall study when the $T$-periodic solution $r_k(t)$ of \eqref{eqddr} coming from Proposition \ref{prop_cont} is of twist type, which implies our result.

The definition of periodic solution of twist type for a differential equation like \eqref{eqddr} is given in \cite{Ortega}. We recall here that a $T$-periodic solution of twist type corresponds to a fixed point of the associated Poincar\'e map, that is accumulated by invariant curves. By the general theory of twist maps and Ordinary Differential Equations (see for example \cite{Mather,Meyer,Moser}) these solutions are Lyapunov stable and accumulated by $(q,p)$-subharmonic and (generalized) quasiperiodic solutions with frequencies $(1,\omega)$ for $p/q$ and $\omega$ small. The definitions of stability, $(q,p)$-subharmonic and (generalized) quasiperiodic solutions of \eqref{eqddr} is a straightforward adjustment of Definition \ref{def_sol}.  

It comes from \cite{Ortega} that the twist character of a periodic solution can be deduced by the third approximation of the equation. To this aim, let us move the solution to the origin via the change of variable $x=r-r_k(t)$ in \eqref{eqddr} and compute the development up to third order around $x=0$. This expression has the form
\[
\ddy+ A_k(t)y +B_k(t)y^2 +C_k(t)y^3 =0.
\]
The result in \cite{Ortega} does not fit in our problem, hence we recall the following version:

\begin{The}[\cite{torres_zhang}]\label{torres_zhang}
	The $T$-periodic solution $r_k(t)$ is of twist type if
	\begin{itemize}
		\item[(i)] $0<(A_k)_*\leq (A_k)^*<\left(\dfrac{\pi}{2T}\right)^2$,
		\item[(ii)] $(C_k)_*>0$,
		\item[(iii)] $10(B_k)_*^2(A_k)_*^{3/2}>9(C_k)^*[(A_k)^*]^{5/2}$,
	\end{itemize}
	where $f^*$ and $f_*$ respectively represent the supremum and infimum of $f\in\mathcal{C}\left([0,T];\mathbb{R}\right)$.
\end{The}
Let us start proving the following:
\begin{Pro}\label{prop_twist}
	Suppose that $I(t)\in\mathcal{C}^4\left([0,T];\mathbb{R}\right)$ and satisfies \eqref{hp_I}. For every admissible strongly non-resonant triplet $(\bar{r},L,p_z)$ there exists $k_1$, such that the solution $r_k(t)$ of \eqref{eqddr} coming from Proposition \ref{prop_cont} is of twist type for $k<k_1$. 
\end{Pro}
\begin{proof}
	As before, let us consider only the case $I_0>0$. Fixing the strongly non resonant triplet $(\bar{r},L,p_z)$ we can apply Theorem \ref{mainA} and consider, for $k<k_0$, the $T$-periodic solutions $r_k(t)$ of \eqref{eqddr}. We recall that   
	$r_k(t) = \bar{r}+\xi_r(t,k)$ where $\xi_r(t,k)$ is a continuous function, $T$-periodic in $t$ and such that $\xi_r(t,k)\rightarrow 0$ as $k\to 0$ uniformly in $t$.
	
	To compute the coefficients $A_k(t),B_k(t),C_k(t)$ of Theorem \ref{torres_zhang}, it is convenient to rewrite equation \eqref{eqddr} as:
	\begin{eqnarray}
	\ddot{r}-\dfrac{L}{r^3}+\dfrac{I_0}{r}(p_z+I_0\ln r)-k\partial_r\left(a(t,r)g(t,r,k)\right)=0,\nonumber
	\end{eqnarray}
	where $g(t,r,k):=p_z+I_0\ln r+ ka(t,r)$. Using this, the calculus is simplified and we obtain the expressions:
	\begin{align*}
	A_k(t) &= \frac{3L^2}{r_k^4(t)}-\frac{I_0p_z}{r_k^2(t)}+\frac{I_0^2(1-\ln r_k(t))}{r_k^2(t)} +k \partial_{rr}\left(a(t,r_k(t))g(t,r_k(t),k)\right); \\
	B_k(t) &= -\frac{6L^2}{r_k^5(t)}+\frac{I_0p_z}{r_k^3(t)}+\frac{I_0^2(2\ln r_k(t)-3)}{2r_k^3(t)} +k  \partial_{rrr}\left(a(t,r_k(t))g(t,r_k(t),k)\right);\\
	C_k(t) &= \frac{10L^2}{r_k^6(t)}-\frac{I_0p_z}{r_k^4(t)}+\frac{I_0^2(11-6\ln r_k(t))}{6r_k^4(t)} +k \partial_{rrrr}\left(a(t,r_k(t))g(t,r_k(t),k)\right).
	\end{align*}
	
	Since the triplet is admissible $I_0p_z=\frac{L^2}{\bar{r}^2}-I_0^2\ln(\bar{r})$, so that we can write the coefficients as follows:
	\begin{equation*}%\label{defABC}
	\begin{split}
	A_k(t) &= \bar{A} + \xi_A (t,r_k(t),k), \qquad  \mbox{with } \bar{A}=\dfrac{2L^2}{\bar{r}^4}+\frac{I_0^2}{\bar{r}^2},
	\\ 
	B_k(t) &= \bar{B}+\xi_B (t,r_k(t),k), \qquad  \mbox{with }\bar{B}=-\dfrac{5L^2}{\bar{r}^5} -\dfrac{3I_0^2}{2\bar{r}^3},
	\\
	C_k(t) &=\bar{C} +\xi_C (t,r_k(t),k), \qquad \mbox{with }\bar{C}=9\dfrac{L^2}{\bar{r}^6}+\dfrac{11}{6}\dfrac{I_0^2}{\bar{r}^4}.
	\end{split}
	\end{equation*}
	Note that, as $r_k(t)$ converges to $\bar{r}$ uniformly in time, then the residual functions $\xi_A, \xi_B, \xi_C$ converge to $0$ uniformly in $t$ as $k\rightarrow 0$. Moreover, $\bar{C}$ is trivially positive, the strong non-resonant condition implies that $\bar{A}<\frac{\pi^2}{4T^2}$ and a direct computation shows that $10\bar{B}^2\bar{A}^{3/2}>9\bar{C}\bar{A}^{5/2}$ (note that it is enough to check that $10\bar{B}^2>9\bar{C}\bar{A}$).
	
	Therefore, using the fact that the remainders $\xi_A, \xi_B, \xi_C$ converge to $0$ uniformly in $t$ as $k\rightarrow 0$, we have that there exists $k_1>0$ such that the coefficients $A_k(t),B_k(t),C_k(t)$ satisfy the hypothesis of Theorem \ref{torres_zhang} for $k<k_1$. 
	
\end{proof}

We are now ready for the
\begin{proof}[Proof of Theorem \ref{mainB}]
	Let us fix an admissible strongly non-resonant triplet
	$(\bar{r},L,p_z)$. From Proposition \ref{prop_twist} we consider
	for every $k<k_0$ the $T$-periodic solution of twist type of
	\eqref{eqddr}. Lyapunov stability is a known property of solutions
	of twist type. Moreover there exist, for $\omega$ and $p/q$ small,
	(generalized) quasi periodic solutions with frequencies $(1,\omega)$
	and solutions $r(t)$ with minimal period $qT$. These periodic
	solutions are such that the function $r(t)-r_k(t)$ has $2p$ zeros in
	a period $[0,qT]$.
	
	Note that, since $r_k(t)\to\bar{r}$ uniformly for $k\to 0$, we also have that $r(t)-\bar{r}$ has $2p$ zeros in a period $[0,qT]$.
	
	Inserting these solutions into system \eqref{eq_cil} we get the thesis, for fixed values of $L$ and $p_z$. In particular we get stability restricted to the integral levels.
	It remains to prove the stability in the whole phase space. We will proceed adapting an argument in \cite{SM}. Consider the map $P(r_0,\dot{r}_0,L,p_z)=(r(T),\dot{r}(T),L,p_z)$ that maps one of these solutions with initial condition $(r_0,\dot{r}_0)$ and integrals $L,p_z$ to the corresponding values at time $T$. To prove stability it is enough to prove that $q_0=(r_0,\dot{r}_0,L_0,p_{z_0})$ is a stable fixed point for the map $P$. Let us fix a neighbourhood $\mathcal{U}$ of $q_0$. Since the point $(r_0,\dot{r}_0)$ is the initial condition of a periodic solution of \eqref{eqddr} of twist type for fixed values of $L,p_z$, we can find in $\mathcal{U}$ an invariant planar region bounded by a closed curve surrounding $(r_0,\dot{r}_0)$ of the form
	\[
	\mathcal{U}_1=\{(r-r_0)^2+(\dot{r}-\dot{r}_0)^2\leq R(\Theta), \: L=L_0, \: p_z=p_{z_0} \}\subset\mathcal{U},
	\]
	where $\Theta=\Theta(r,\dot{r})$ represents the angle centred in $(r_0,\dot{r}_0)$.
	Equation \eqref{eqddr} depends continuously on the parameters $L,p_z$, hence by continuous dependence, we can find a family of curves
	\[
	(r-r_0)^2+(\dot{r}-\dot{r}_0)^2\leq R(\Theta,L,p_z)
	\]
	depending continuously on $L,p_z$ and with the properties above. Note that here we have used the fact that the strongly non resonant condition \eqref{stronglynonresonant} is an open condition (for fixed $\bar{r}$). Therefore, for sufficiently small $\delta$, the region
	\[
	\mathcal{U}_2=\{(r-r_0)^2+(\dot{r}-\dot{r}_0)^2\leq R(\Theta,L,p_z), \: |L-L_0|<\delta, \: |p_z-p_{z_0}|<\delta \}
	\]
	is invariant and contained in $\mathcal{U}$. Therefore, the point $q_0$ is stable under the map $P$ and the solutions are radially stable.
\end{proof}

\section{Conclusions}
In this paper we described the
non-relativistic dynamics induced by a periodically time dependent current along an infinite long straight wire, which is ruled by the Newton-Lorentz equation \eqref{eqLintro}. The constant case represents a classical magnetostatic situation where every charged particle with non-null angular momentum is radially periodic and definitely increasing in the current direction. Moreover, the motions are radially stable and then each particle is confined between two cylinders. The introduction of a time dependent perturbation in the current generates an electric field which breaks the magnetostatic structure of the dynamics. According to this, Biot-Savart law does not hold and it is necessary to deduce the electromagnetic field by solving Maxwell equations in a distributional sense. However, this cannot be done directly because the current distribution is not of compact support so that we considered a sequence of approximated problems to deduce it properly.

We proved that non resonant radially periodical motions are preserved when the time dependent perturbation is small enough. This solution depends on the fixed values of angular and linear momenta, which are first integrals of motions. To this aim we use the local continuation Theorem to the equilibrium of the constant current dynamical system, where the resonances must be excluded. Furthermore, a stability analysis is performed for the perturbed dynamics. In particular, we applied the third approximation method to prove that the radial component of the just found radially periodic solutions are of twist type, implying radially Lyapunov stability for fixed values of the momenta. Finally, we extended Lyapunov stability to the phase plane adapting an argument of \cite{SM}, i.e. stability still holds for small variations of the momenta. 

In conclusion, we gave an analytical study of the dynamics and some stability properties induced by a time dependent current. This represents a counterpart of the well known results for constant currents. We believe that our techniques could be adapted to the study of different configurations of the wire, such as the case of a circular wire.

\section*{Acknowledgements}
The authors are grateful to Pedro Torres and Rafael Ortega for fruitful discussions and suggestions.

\end{document}